\documentclass[aps,floatfix,letterpaper,twocolumn,nofootinbib,superscriptaddress]{revtex4}
\usepackage{geometry} 
\UseRawInputEncoding
\usepackage[caption=false]{subfig}
\usepackage{graphicx} 
\usepackage{epstopdf}

\usepackage{array} 
\usepackage{verbatim} 
\usepackage{amsmath,amsfonts,amssymb,amscd}
\usepackage{amsthm}
\usepackage{tabularx}

\usepackage{hyperref}
\hypersetup{
    bookmarksnumbered=true, 
    unicode=false, 
    pdfstartview={FitH}, 
    pdftitle={}, 
    pdfauthor={}, 
    pdfsubject={}, 
    pdfcreator={}, 
    pdfproducer={}, 
    pdfkeywords={}, 
    pdfnewwindow=true, 
    colorlinks=true, 
    linkcolor=blue, 
    citecolor=blue, 
    filecolor=blue, 
    urlcolor=blue 
}

\newtheorem{thm}{Theorem}[section]

\newtheorem{lem}[thm]{Lemma}

\numberwithin{equation}{section}

\newcommand{\eq}[1]{(\hyperref[eq:#1]{\ref*{eq:#1}})}

\renewcommand{\sec}[1]{\hyperref[sec:#1]{Section~\ref*{sec:#1}}}
\newcommand{\thrm}[1]{\hyperref[thm:#1]{Theorem~\ref*{thm:#1}}}
\newcommand{\lemm}[1]{\hyperref[lemm:#1]{Lemma~\ref*{lemm:#1}}}
\newcommand{\prop}[1]{\hyperref[prop:#1]{Proposition~\ref*{prop:#1}}}
\newcommand{\corr}[1]{\hyperref[corr:#1]{Corollary~\ref*{corr:#1}}}
\newcommand{\fig}[1]{\hyperref[fig:#1]{Figure~\ref*{fig:#1}}}

\newcommand{\ket}[1]{|#1\rangle}
\newcommand{\bra}[1]{\langle#1|}

\newcommand{\tth}[0]{\textsuperscript{th}}

\usepackage[usenames,dvipsnames]{color}
\usepackage{xcolor}

\DeclareMathAlphabet{\matheu}{U}{eus}{m}{n}

\DeclareMathOperator{\tr}{tr}

\newcommand{\sop}[1]{{\mathcal #1}}

\newcommand{\I}{{\mathbb I}}
\newcommand{\X}{{\mathbb P}_X}
\newcommand{\Y}{{\mathbb P}_Y}
\newcommand{\Z}{{\mathbb P}_Z}

\newcommand{\ketbra}[2]{|{#1}\rangle\!\langle{#2}|}

\newcommand{\no}{\nonumber\\}

\newcolumntype{L}[1]{>{\raggedright}p{#1}}
\newcolumntype{C}[1]{>{\centering}p{#1}}
\newcolumntype{R}[1]{>{\raggedleft}p{#1}}
\newcolumntype{D}{>{\centering\arraybackslash}X}

\begin{document}

\title{Robust Calibration of a Universal Single-Qubit Gate-Set via Robust Phase Estimation}
\author{Shelby Kimmel}
\affiliation{Joint Center for Quantum Information and Computer Science (QuICS), University of Maryland,
College Park, MD 20742}
\affiliation{Center for Theoretical Physics, Massachusetts Institute of Technology, Cambridge, MA 02139}
\author{Guang Hao Low}
\author{Theodore J. Yoder}
\affiliation{Department of Physics, Massachusetts Institute of Technology, Cambridge, MA 02139}

\begin{abstract}
An important step in building a quantum computer is calibrating
experimentally implemented quantum gates to produce operations that
are close to ideal unitaries. The calibration step involves estimating
the systematic errors in gates and then using controls to correct the
implementation. Quantum process tomography is a standard technique for
estimating these errors, but is both time consuming, (when one only
wants to learn a few key parameters), and is usually inaccurate without resources like
perfect state preparation and measurement, which might not be available.
With the goal of efficiently and accurately estimating specific errors using minimal
resources, we develop a parameter estimation technique, which can
gauge key systematic parameters (specifically, amplitude and off-resonance errors) in 
a universal single-qubit gate-set with provable robustness and efficiency. In particular, 
our estimates achieve the optimal efficiency, Heisenberg scaling, and do so without 
entanglement and entirely within a single-qubit Hilbert space. Our main theorem making 
this possible is a robust version of the phase estimation procedure of Higgins et al. \cite{HBB+7}.
\end{abstract}

\maketitle

\noindent{\underline{\Large{\textbf{Errata:}}} In Eq. (\ref{eq:p_error}), we state that the probability of an error occuring at the $j$th iteration is given by the probability that the estimate is outside of $\pi/(2k_j)$ of the actual value. However, as pointed out in Refs.
\cite{PhysRevA.102.042613,PhysRevA.103.042609}, an error can in fact occur
even when the estimate at the $j$th round is within of $\pi/(2k_j)$ of the true value. The correct condition that characterizes when an error occurs is that the estimate is outside of $\pi/(3k_j)$ of the actual value \cite{PhysRevA.102.042613,PhysRevA.103.042609}. Thus, while the detailed analysis of App. \ref{app:bounding} does bound the probability that the $j$th iteration estimate is within $\pi/(2k_j)$ of the true value, this only gives a lower bound on the probability of an error occuring. 

In fact, Ref. \cite{HBB+7} already bounds the probability that the estimate at
the $j$th round is within $\pi/(3k_j)$ of the actual value and Ref.
\cite{PhysRevA.102.042613} provides a tighter analysis - so while our analysis in Sec. \ref{sec:WE} is not correct, Refs. \cite{HBB+7,PhysRevA.102.042613} show that the general procedure still attains Heisenberg scaling, with a slightly worse constant overhead than the performance claimed in this work. A simple argument in
App. B of \cite{PhysRevA.103.042609} then shows how to leverage those analyses to obtain (through an additional constant overhead) a procedure with Heisenberg scaling that is robust to additive errors, superceding 
Sec. \ref{sec:AE} and showing the main result of this work - that phase estimation can provide robust calibration of a universal single gate set - is still valid. 

\section{Introduction}

Not all errors in a quantum computation experiment are created equal. There are actually two broad
classes of error, unitary errors, also known as systematic errors, and nonunitary errors, also 
known as decoherence. Both sets of errors need to be corrected below a certain threshold
for scalable quantum computation to take place \cite{Aharonov1997,Preskill1998}. Correcting systematic errors, such as
over-rotation or off-resonance errors, is typically regarded as the easier task; because these 
errors are directly related to the controls available to an experimenter, they can be directly corrected
by changing those controls. In this respect systematic errors contrast with decoherence, which 
is typically less affected by an experimenter's control software and more influenced by imperfect 
or nonideal hardware.

However, even though systematic errors are considered the easier of the two to correct, calibrating 
gates in a quantum computer to reduce systematic errors can still take hours even for modest system 
sizes, and moreover this calibration may have to be repeated every time the quantum computer is 
switched on \cite{Lucero14}. Not only can this process be inefficient in terms of the precision of the estimates
with respect to time, but standard techniques for estimating systematic errors
often suffer from measurement bias, leading to inaccurate estimates \cite{CGT+09}.

To characterize systematic errors, quantum process tomography \cite{CN97}
has long been a valuable tool in the experimental toolkit. However,
standard techniques \cite{CN97} require perfect state preparation,
perfect measurement, and at least some perfect gates. Especially
during the calibration stage of an experiment, it is unreasonable to assume access to
such perfect resources, and, without them, standard process tomography results in a difficult nonlinear
estimation problem \cite{stark1,stark3,stark2014}, and hence the estimates obtained using
this technique are typically inaccurate. Moreover, systematic errors are controlled
by a few key parameters, but unless
the measurement basis of the tomography procedure is specialized,
e.g. \cite{Bendersky2009}, to extract those few important parameters
can require resources that scale exponentially with the size of the system
and can be time consuming even for single qubit processes.

Recent approaches aim to circumvent the stringent requirements of
standard tomography. Randomized benchmarking (RB)
\cite{KLR+08,MGE11,MGJ+12}, randomized benchmarking tomography
(RBT) \cite{kimmel14}, and other tools based
on randomized benchmarking \cite{WGHF15,WBE14} can characterize quantum error processes even when
nothing is known about state preparation and measurement. However,
these procedures require access to relatively good Clifford operations
\cite{MGE12,ECM+13}. In addition, other than certain key parameters like the average
fidelity, single parameters cannot be
extracted efficiently. While the average fidelity can be learned
efficiently using RB, average fidelity gives no information about the nature of the
systematic errors on the gates, and so is useless for experimentalists
who would like to use tomographic data to correct systematic errors.

Another promising approach is gate-set tomography (GST) \cite{BKG+13,MGS+12}.
GST makes no assumptions about state preparation, measurement, or
processes, while still obtaining accurate estimates. However, GST is even more inefficient than
standard tomography, since to learn even a single parameter, one must
fully characterize a complete gate-set along with state preparation
and measurement.

We propose a new procedure to estimate simultaneously all the systematic errors in a universal single-qubit gate-set.
This procedure falls in between existing protocols in terms of required
resources and assumptions, but is optimal in terms of asymptotic efficiency. 
Rather than doing full tomography, we extract only parameters that correspond to systematic errors,
precisely the errors that the experimentalists can easily correct. We learn those
parameters efficiently and non-adaptively --- in fact we are Heisenberg
limited. Like GST, we require no perfect resources, and, moreover, we do not require any additional gates besides the ones
we are characterizing. We also never require more than a single-qubit Hilbert space. In particular, we never need entangled 
states, like those often employed in interferometric phase estimation procedures \cite{Kacprowicz2010, Vidrighin2014}.
Instead, the source of the quantum advantage in our procedure is the exploitation of long coherence times of the qubit
system, and our ability to apply a gate multiple times in series. This allows small variations in gates
 to coherently accumulate into large observables.

Of course, like other Heisenberg limited studies \cite{Huelga1997, Kolodynski2010}, 
a finite coherence time ultimately limits the estimation accuracy that we can achieve. 
However, our procedure does retain Heisenberg scaling against 
state preparation errors and measurement errors.  Thus, while a standard parameter 
estimation scheme (one that repeatedly prepares a state, applies an operation, and then 
measures) is limited by uncertainty in the measurement operator, our procedure can obtain
Heisenberg-limited, arbitrarily precise parameter estimates even with unknown (but not too
large) errors in the measurement operator. In this way, our procedure also has
some of the flavor of randomized benchmarking.

In order to achieve these gains in efficiency and accuracy, we lose some
of the flexibility of other procedures. Our procedure will fail if errors
are larger than some threshold amount. Also, the procedure is most useful
when the experimentalist has precise control over the gates, and can undo
the systematic errors once they are characterized. We hope that a calibration procedure
like the one we describe could be used to quickly ``tune up'' gates
before more sophisticated procedures like RBT or GST are employed
to characterize non-systematic (decoherence) errors.

Our main theorem says that it is possible to perform phase estimation
in the presence of errors. In particular, we consider additive errors
in the measurement probabilities of experiments. This is a fairly
straightforward idea, but it turns out that many different effects can
be swept into these additive errors. For example, state preparation and
measurement errors can be seen as additive errors. We show how to
do phase estimation in the presence of these additive errors and extract two parameters 
of a process, amplitude and off-resonance errors, instead of only learning the phase of 
a rotation, as is typical. It turns out that while estimating one of the parameters of interest, 
the effect of the other parameter can be thought of as another additive error. Moreover, 
when multiple additive errors occur simultaneously, the result is still an additive error, 
with (worst-case) magnitude equal to the sum of the magnitudes of the individual additive errors.

In particular, we modify and improve a non-adaptive phase estimation technique of
Higgins et al. \cite{HBB+7} to show
\begin{thm}\label{main_theorem}
Suppose that we can perform two families of experiments, $\ket{0}$-experiments
and $\ket{+}$-experiments, indexed by $k\in\mathbb{Z}^+$, whose
probabilities of success are, respectively,
\begin{align}
p_0(A,k)&=\frac{1+\cos(kA)}{2}+\delta_0(k),\label{eq:p0_orig}\\
p_+(A,k)&=\frac{1+\sin(kA)}{2}+\delta_+(k).\label{eq:pp_orig}
\end{align}
Also assume that performing either of the $k^{\text{th}}$ experiments
takes time proportional to $k$, and that
\begin{align}
\sup_k\left\{|\delta_0(k)|,|\delta_+(k)|\right\}< 1/\sqrt{8}.
\end{align}
Then an
estimate $\hat{A}$ of $A\in(-\pi,\pi]$ with standard deviation $\sigma(\hat A)$ can be
obtained in time $T=\mathcal{O}(1/\sigma(\hat A))$ using non-adaptive
experiments.

On the other hand, if $|\delta_0(k)|$ and
$|\delta_+(k)|$ are less than $1/\sqrt{8}$ for all $k<k^*$, then
it is possible to obtain an estimate $\hat A$ of $A$ with $\sigma(\hat A)\sim O(1/k^*)$
(with no promise on the scaling of the procedure).
\end{thm}
More precise bounds on the scaling of standard deviation with
time can be found in Section \ref{sec:phase_estimation}.

We call the terms $\delta_{0}(k)$ and $\delta_{+}(k)$  {\it{additive errors}}.
While we can only achieve Heisenberg scaling up to arbitrary precision
when the additive errors have magnitude less than
$1/\sqrt{8}$ for all $k$, some effects  (like depolarizing errors)
cause additive errors that grow with $k$ and so eventually
overwhelm the $1/\sqrt{8}$ bound. However, in that case, if $k^*$
is the $k$ at which the errors become too large, our procedure can
give an estimate with precision that is $\mathcal O(1/k^*)$, which is often
 better than standard procedures which are limited by uncertainty
 in state preparation and measurement.

The layout of the paper is as follows. First, in Section \ref{sec:setup}, we
define notation for single qubit operations and errors. In Section
\ref{sec:sequences} we use Theorem \ref{main_theorem} to calibrate systematic errors in a single-qubit gate-set,
and then Section \ref{sec:errors} discusses the robustness of this procedure to
sources of error such as imperfect state preparation,
measurement noise, and decoherence. Finally, in Section \ref{sec:phase_estimation} 
we modify and reanalyze the non-adaptive Heisenberg limited phase estimation 
procedure of \cite{HBB+7} to achieve better scaling and simpler bounds, resulting in
the proof of Theorem \ref{main_theorem}.

\section{Characterizing a Universal Gate-Set}\label{sec:setup}

We consider systematic errors in a universal single-qubit 
gate-set. For the moment, we assume that the implemented gates 
have systematic errors but no decoherence errors, and hence are perfect unitaries.
(We relax these assumptions in Section \ref{sec:errors}.)
Single-qubit unitaries are defined by two parameters: their axis of rotation
and their angle of rotation in the Bloch sphere. (See \cite{Nielsen2004} 
for background on the Bloch sphere.) 

Two unitary gates are sufficient to create a universal single-qubit
gate-set. We describe a scheme to characterize a gate-set where
the two gates are ideally orthogonal. In particular, we consider the
case that one gate is a faulty implementation of $Z_{\pi/2}$, a $\pi/2$
rotation about the $Z$-axis of the Bloch sphere, and the other gate
is a faulty implementation of $X_{\pi/4}$, a $\pi/4$ rotation about the
$X$-axis. We also assume that the experimenter can create an
imperfect $\ket{0}$ state, the 1-valued eigenstate of $Z_{\pi/2}$. How good
the gates, state preparation, and measurement must be initially for our
procedure to work is determined by Theorem \ref{main_theorem}, and 
will be made clear in the calibration procedures in Section \ref{sec:sequences}.

We chose specific rotation angles for our $Z$ and $X$
rotations. This choice is mainly for convenience, since it turns out 
that access to (1) imperfect versions of the states
\begin{align}
\ket{0},&&
\ket{+}=\frac{\ket{0}+\ket{1}}{\sqrt{2}},&&
\ket{\!\!\rightarrow}=\frac{\ket{0}+i\ket{1}}{\sqrt{2}},
\end{align}
and to (2) a $Z_{\pi}$ rotation calibrated to near perfection,
are sufficient to characterize $Z_{\chi}$ and $X_{\phi}$ for any rotations
$\chi$ and $\phi$ using our techniques. Only calibration of $X_{\phi}$ requires the second
condition. These two conditions are satisfied given the
gate-set in the previous paragraph. Indeed, in an experiment where
$Z_{\chi}$ and $X_{\phi}$ are available, albeit erroneously, for any
$\chi$ and $\phi$, it would perhaps be best to first calibrate $Z_{\pi/2}$
and $X_{\pi/4}$ rotations so that conditions (1) and (2) are satisfied before
calibrating $Z_{\chi}$ and $X_{\phi}$ for arbitrary $\chi$ and $\phi$.

We now define our universal gates mathematically.
Without loss of generality, we can define the $Z$-axis of the Bloch
sphere to be aligned with the axis of rotation of our approximate $Z_{\pi/2}$
gate. This means that our initial state preparation may not be
aligned with the $Z$-axis, but our scheme is robust against this
type of error. Once the axis of our approximate $Z_{\pi/2}$
gate is fixed to the $Z$-axis, the only free parameter is the angle
of rotation. Thus, we can write our approximate $Z_{\pi/2}$ gate as
\begin{align}\label{eq:Zrot}
Z_{\pi/2}(\alpha)=\cos\left(\frac{\pi}{4}(1+\alpha)\right)\I-i\sin\left(\frac{\pi}{4}(1+\alpha)\right)\Z,
\end{align}
where $\{\X,\Y,\Z\}$ are the Pauli matrices, $\I$ is the $2\times2$
identity matrix, and $\alpha$ is a parameter that quantifies how far
the implemented angle of rotation is from $\pi/2$. When
$\alpha=0$, we have implemented a perfect gate. 

Likewise, without loss of generality, we
define the $X$-axis of the Bloch sphere so that the axis
of rotation of our approximate
$X_{\pi/4}$ gate lies along the $XZ$-plane of the Bloch sphere. In this
case, the approximate $X_{\pi/4}$ gate has two degrees of freedom: the location of the
axis of rotation in the $XZ$-plane of the Bloch sphere, and its angle 
of rotation. More precisely, we can write our approximate $X_{\pi/4}$ gate as
\begin{align}
\label{eq:R1}
X_{\pi/4}(\epsilon,\theta)=&\cos\left(\frac{\pi}{8}
\left(1+\epsilon\right)\right)\I-i\sin\left(\frac{\pi}{8}
\left(1+\epsilon\right)\right)
\no&\times\left(\cos(\theta)\X+\sin(\theta)\Z\right),
 \end{align}
 where $\theta$ is the angle of the axis of rotation relative
 to the $X$-axis, and $\epsilon$ is a parameter
 that quantifies how far
the implemented angle of rotation is from $\pi/4$.
When $\epsilon=\theta=0$ we have implemented a perfect gate.

Our goal is to estimate $\alpha$, $\theta$, and $\epsilon$, with the
expectation that once these systematic errors have been quantified,
experimentalists can adjust the controls of the gates to set their
values close to $0$. If desired, the process can then be repeated --
the new values of $\alpha$, $\theta$, and $\epsilon$ can be reestimated 
and readjusted again.

We will also need notation for a general imperfect $X$ rotation:
\begin{align}
\label{eq:R2}
X_{\phi}(\epsilon,\theta)=&\cos\left(\frac{\phi}{2}
\left(1+\epsilon\right)\right)\I-i\sin\left(\frac{\phi}{2}
\left(1+\epsilon\right)\right)
\no&\times\left(\cos(\theta)\X+\sin(\theta)\Z\right).
 \end{align}
This expression $ X_{\phi}(\epsilon,\theta)$ represents a rotation that is in 
the $XZ$ plane of the Bloch sphere, which is approximately a rotation
by an angle $\phi.$ In general, the parameters $\epsilon$ and $\theta$ will
depend on $\phi$.

In some cases, we will apply the unitary operations $X_{\pi/4}(\epsilon,\theta)$
and $Z_{\pi/2}(\alpha)$ to mixed states instead of pure states. In this case, 
we will use cursive letters to represent the CPTP maps corresponding to these unitaries. That is
\begin{align}
\sop X_{\pi/4}(\epsilon,\theta)(\rho)=
X_{\pi/4}(\epsilon,\theta)\rho\left(X_{\pi/4}(\epsilon,\theta)\right)^\dagger,\nonumber\\
\sop Z_{\pi/2}(\alpha)(\rho)=
Z_{\pi/2}(\alpha)\rho\left(Z_{\pi/2}(\alpha)\right)^\dagger,
\end{align}
where $\dagger$ denotes the conjugate transpose.

We use the notation $Z_{\pi/2}(\alpha)^k$ to mean $k$ repeated applications of $Z_{\pi/2}(\alpha).$
Unitaries act right to left, so
$Z_{\pi/2}(\alpha)X_{\pi/4}(\epsilon,\theta)$ means apply the $X$-rotation first, and then
the $Z$-rotation.

\section{Sequences for Estimating Systematic Errors}\label{sec:sequences}

In this section, we describe sequences consisting of unitaries $Z_{\pi/2}(\alpha)$
and $X_{\pi/4}(\epsilon,\theta)$, which can be used to estimate
the systematic error parameters $\alpha$, $\epsilon,$ and $\theta.$
In particular, we would like to obtain observables $p_0(\alpha,k)$, $p_+(\alpha,k)$, $p_0(\theta,k)$, $p_+(\theta,k)$,
$p_0(\epsilon,k)$, and $p_+(\epsilon,k)$, as described in Theorem
\ref{main_theorem}. By Theorem \ref{main_theorem}, such observables
will allow us to accurately estimate $\alpha$, $\epsilon$ and $\theta$ as long as
the additive errors associated with these observables are not too large.
We address the problem of initially bounding additive errors in Appendix \ref{app:bounding}.

In this section, we assume that we can prepare the states $\ket{0}$, $\ket{+}$, and $\ket{\!\!\rightarrow}$
perfectly, and that we can measure (perfectly) the probability of being in the state  $\ket{0}$,
or the probability of being in the state $\ket{+}$. In Section \ref{sec:errors}, we introduce state 
preparation and measurement errors to our protocols.

\subsection{Estimating $\alpha$}\label{sub:alpha}

With the assumption of perfect state preparation and measurement,
we can estimate $\alpha$ using standard phase estimation, without having to resort to robust
phase estimation. One can verify that
\begin{align}
|\bra{+}Z_{\pi/2}(\alpha)^k\ket{+}|^2&=\frac{1+\cos\left(-k\frac{\pi}{2}(1+\alpha)\right)}{2},\nonumber\\
|\bra{+}Z_{\pi/2}(\alpha)^k\ket{\rightarrow}|^2&=\frac{1+\sin\left(-k\frac{\pi}{2}(1+\alpha)\right)}{2}.
\end{align}
Comparing with Eqs. \ref{eq:p0_orig} and \ref{eq:pp_orig}, we see
these sequences can be used to estimate $\alpha$. If $N$ is the number
of times we apply $Z_{\pi/2}(\alpha)$, by Theorem
\ref{main_theorem}, we can obtain an estimate of $\alpha$ with with
standard deviation  $\sop O(1/N)$. This is what is meant by Heisenberg
scaling or Heisenberg limited. ($N$ is the most natural and unambiguous measure of resource consumption
for phase estimation; see the appendix of \cite{HBB+7}). 

\subsection{Estimating $\epsilon$}\label{sub:epsilon}

We next describe the sequences used to estimate $\epsilon$. In this
section, for ease of explication later in the paper, we will
characterize the general gate $X_{\phi}(\epsilon,\theta)$, where
we can always substitute $\pi/4$ for the variable $\phi$ to obtain
the results relevant to $X_{\pi/4}(\epsilon,\theta).$ Let 
$\phi_\epsilon=\phi(1+\epsilon)$.
Again, a simple calculation shows that 
\begin{align}\label{eq:exp0_ep}
|\bra{0}X_{\phi}(\epsilon,\theta)^k\ket{0}|^2=
&\frac{1+\cos\left(k\phi_\epsilon\right)}{2}
+\sin^2\left(\frac{k\phi_\epsilon}{2}\right)\sin^2\theta,\nonumber\\
|\bra{0}X_{\phi}(\epsilon,\theta)^k\ket{\rightarrow}|^2=
&\frac{1+\sin\left(k\phi_\epsilon\right)}{2}
-\sin\left(k\phi_\epsilon\right)\sin^2\frac{\theta}{2}.
\end{align}
Comparing with Eq. \eq{p0_orig}, we see this sequence allows us
make a measurements with success probabilities $p_{0/+}(\phi_\epsilon,k)$,
with $|\delta_0(k)|,|\delta_+(k)|\le\sin^2(\theta)$.

By Theorem \ref{main_theorem}, as long as $|\theta|$ is less than
about $36^{\circ}$, (along with our current assumptions of perfect
state preparation and measurement) then we can estimate
$\phi_\epsilon$, and hence $\epsilon$ (assuming a constant $\phi$),
with standard deviation $\sop O(1/N)$, where
$N$ is the total number of times $X_{\phi}(\epsilon,\theta)$ is
used over the course of the protocol.

In Appendix \ref{app:bounding},
we show how to independently bound the size of $\theta$, in order to
determine if $|\theta|$ is small enough to apply this protocol.

\subsection{Estimating $\theta$}\label{sub:theta}

We now discuss sequences to estimate $\theta.$ For the moment,
we assume that after estimating $\alpha$, we are able to
set $\alpha=0$ exactly. In Section \ref{sec:nonzero_alpha} we will examine what
happens to this protocol when $\alpha$ is not zero.

Consider the rotation
\begin{align}\label{eq:U_seq}
U&=Z_{\pi/2}(0)X_{\pi/4}(\epsilon,\theta)^4
Z_{\pi/2}(0)^2X_{\pi/4}(\epsilon,\theta)^4Z_{\pi/2}(0).
\end{align}
Then, because any single-qubit unitary can be written as a rotation of some
angle $\Phi$ about an axis $\vec{n}$ in the Bloch sphere, we may write
\begin{align}
U=\cos{\left(\frac{\Phi}{2}\right)}\mathbb{I}-
i\sin{\left(\frac{\Phi}{2}\right)}
\vec{n}\cdot(\X,\Y,\Z).
\end{align}

By direct expansion, we find that the $Y$-component of $\vec{n}$ is zero and
\begin{align}
n_X&=-\frac{\cos(\theta)\cos\left(\frac{\pi\epsilon}{2}\right)}
{\sqrt{1-\sin^2\theta\cos^2\left(\frac{\pi\epsilon}{2}\right)}},\\
n_Z&=\frac{\sin\left(\frac{\pi\epsilon}{2}\right)}
{\sqrt{1-\sin^2\theta\cos^2\left(\frac{\pi\epsilon}{2}\right)}},\\
\sin\left(\frac{\Phi}{2}\right)&=2\sin(\theta)\cos\left(\frac{\pi\epsilon}{2}\right)
\sqrt{1-\sin^2(\theta)\cos^2\left(\frac{\pi\epsilon}{2}\right)}.
\end{align}
We define the angle $\Theta$ to be such that $\cos(\Theta)=n_X$ and
$\sin(\Theta)=n_Z$. Using our notation of Section
\ref{sec:setup}, we may write $U=X_{\Phi}(0,\Theta)$, and hence,
using the techniques of Section \ref{sub:epsilon}, we can obtain a Heisenberg
limited estimate of $\Phi$ as long as $|\Theta|$ is not too large.  
All that remains is to show that an estimate of $\Phi$ allows us to 
estimate $\theta$ with similar precision, and that $\Theta$ is
not too large.  

We have
\begin{align}
|\Theta|&=\arcsin|n_Z|\no
&=\arcsin\left|\frac{\sin(\pi\epsilon/2)}{\sqrt{1-\sin^2\theta\cos^2(\frac{\pi\epsilon}{2})}}\right|,
\end{align}
which implies $\sin^2\Theta$ scales as $\mathcal{O}(\epsilon^2)$. In particular, if
$\sin^2\theta<1/\sqrt{8}$, as is necessary for estimating $\epsilon$ using the
methods of Section \ref{sub:epsilon}, then $|\epsilon|< 0.341$ is sufficient for 
estimating $\Phi$. We can independently verify whether 
$|\epsilon|$ is small enough for the protocol to succeed
using the techniques of Appendix \ref{app:bounding}.

We now show that estimating $\Phi$ is sufficient to estimate $\theta$.
We have
\begin{align}\label{eq:Btotheta}
\sin{\frac{\Phi}{2}}=&2\sin{\theta}\cos{\frac{\pi\epsilon}{2}}
\sqrt{1-\sin^2{\theta}\cos^2{\frac{\pi\epsilon}{2}}},
\end{align}
which can be expanded, assuming small $\theta$, as
\begin{align}
\sin{\frac{\Phi}{2}}=2\theta\cos\frac{\pi\epsilon}{2}+\mathcal{O}(\theta^3).
\end{align}
Since $\epsilon$ can be estimated from Section \ref{sub:epsilon}, we can estimate
\begin{align}\label{eq:thetaFromA}
\theta=\frac{\sin(\Phi/2)}{2\cos(\pi\epsilon/2)}.
\end{align}
As long as $\epsilon$ and $\theta$ are not too large, the relationship
between $\Phi$ and $\theta$ is very close to linear, so if we know
 the standard deviation $\sigma(\hat{\Phi})$ of $\hat \Phi$, our estimate of $\Phi$,
 we can obtain the standard deviation of our estimate of $\theta$, $\sigma(\hat{\theta})$ as
 \begin{align}
 \sigma(\hat \theta)\le\frac{\sigma(\hat \Phi)}{4\cos(\pi\epsilon/2)}.
 \end{align}
 Since we can estimate $\Phi$ with Heisenberg limited uncertainty,
 this means we can estimate $\theta$ with Heisenberg limited uncertainty.

In the case that the relationship between $\theta$ and $\Phi$ is not
close to linear (which can be checked using Eq. \eq{Btotheta})
then while our technique gives a bound on the variance of our
estimate of $\Phi$, because we don't know the form of the distribution
of this estimate, we can not easily bound the variance of our estimate
of $\theta.$ In this case, we recommend using non-parametric bootstrapping
\cite{ET94}, which, at the cost of a constant multiplicative overhead, can be used
to estimate the variance of the estimate of $\theta$ obtained from this
procedure, without any assumptions on a linear relationship between
$\Phi$ and $\theta.$ While it is possible that this non-linearity would
keep the estimate of $\theta$ from being Heisenberg limited,
as long as the variance of our estimate of
$\Phi$ is small, the relationship between $\Phi$ and $\theta$
should be approximately linear, and
so we expect that we will always be Heisenberg limited in our estimate.

In Section \ref{sec:setup}, we claimed that our techniques can
be applied to characterize $Z_{\chi}(\alpha)$ and 
$X_{\phi}(\epsilon,\theta)$ for arbitrary $\chi$ and $\phi$. Our techniques immediately
extend to give estimates of $\alpha$ and $\epsilon$ for these rotations, but
it may not be immediately clear how to obtain an estimate of $\theta$ in this case. 
The procedure is quite straightforward. First, choose
a positive integer $q$ such that $q\phi=t\pi$ for an odd integer $t$.\footnote{It may happen such a $q$ is impossible
to find (e.g. if $\phi=2\pi/3$). Such cases occur when $\phi=(a/b)\pi$ for $a/b$ a reduced fraction and 
$a$ even. However, letting $c=a/2^s$ be the odd integer part of $a$, calibrating a rotation by $\phi'=(c/b)\pi$ 
is possible, and a rotation by $\phi$ can be obtained by doing $2^s$ rotations by $\phi'$.}
Construct
\begin{align}
U_\phi&=Z_{\pi/2}(0)X_{\phi}(\epsilon,\theta)^qZ_{\pi/2}(0)^2X_{\phi}(\epsilon,\theta)^qZ_{\pi/2}(0).
\end{align}
Using the same procedure as before, we can then estimate $\theta$, assuming $|t\epsilon|$ is not
too large ($|t\epsilon|<0.341$ is sufficient if $\sin^2(\theta)<1/\sqrt{8}$).

\section{Bounding and Quantifying Other Errors} \label{sec:errors}

In Section \ref{sec:sequences}, we showed how to construct sequences
such that, if states are prepared perfectly, measurements are performed perfectly,
and the gates are exactly of the form we assume, then one can estimate $\alpha$,
$\epsilon$, and $\theta$ at the  Heisenberg limit.
In this section, we  show that these assumptions can be relaxed, 
and examine their effect on our protocol.

We will completely restrict ourselves to a Hilbert space of
dimension 2. (So we assume all states and operators exist and act
only on this subspace.) Let $Pos(2)$ be the set of positive semidefinite operators on the 
Hilbert space of dimension 2. By $A\succeq B$, we mean $A-B$
is positive semidefinite.   
Consider a general scenario in which we would like to prepare a state
$\rho$, apply a CPTP map $\sop E$ (which might be a sequence of 
gates), and then measure with the POVM 
$\sop W=\{W_1,\dots,W_k\}$. Then the probability of obtaining outcome
$i$ is 
\begin{align}
p_i=\tr\left(W_i\sop E(\rho)\right).
\end{align}
Suppose, however, that instead of preparing the state $\rho$ perfectly,
we prepare the faulty state $\rho'$, apply the faulty CPTP map $\sop E'$
and measure using the faulty POVM $\sop W'=\{W_1',\dots,W_k'\}$. In this
case, the probability of obtaining outcome $i$ is
\begin{align}
p_i'=\tr\left(W_i'\sop E'(\rho')\right).
\end{align}
Since we care about additive errors, which are a difference in probability
between the desired experiment and the implemented experiment, we
would like to bound $|p_i-p_i'|.$

Using the triangle inequality, we have
\begin{align}\label{eq:triangle_errors}
|p_i-p_i'|=&|\tr\left(W_i\sop E(\rho)\right)-\tr\left(W_i\sop E'(\rho)\right)|\nonumber\\
&+|\tr\left(W_i\sop E'(\rho)\right)-\tr\left(W_i'\sop E'(\rho)\right)|\nonumber\\
&+|\tr\left(W_i'\sop E'(\rho)\right)-\tr\left(W_i'\sop E'(\rho')\right)|.
\end{align}
Thus the difference in experimental outcome can be split into 
separate contributions due to gate error, measurement error,  and 
state preparation error.

In particular, measurement error is bounded by
\begin{align}
\delta_{W_i,W_i'}\equiv\max_{\substack{\rho\in Pos(2)\\\tr(\rho)=1}}|\tr\left((W_i-W_i')\rho\right)|,
\end{align}
state preparation error is bounded by
\begin{align}
\delta_{\rho,\rho'}\equiv\max_{\substack{W\in Pos(2)\\W\preceq \I}}
|\tr(W(\rho-\rho')|=\frac{1}{2}\|\rho-\rho'\|_1,
\end{align}
where $\|\cdot\|_1$ is the $l_1$ norm or ``trace distance" (see \cite{Nielsen2004}), and the gate
error is bounded by\footnote{We use the bounded rather than completely bounded (diamond)
norm here because we are restricting our Hilbert space to be of dimension 2.}
\begin{align}
\delta_{\sop E,\sop E'}&\equiv\max_{\substack{W,\rho\in Pos(2) \\ W\preceq \I\\\tr(\rho)=1}}
|\tr\left(W\sop E(\rho)\right)-\tr\left(W\sop E'(\rho)\right)|\nonumber\\
&=\frac{1}{2}\max_{\substack{\rho\in Pos(2)\\\tr(\rho)=1}}\|\sop E(\rho)-\sop E'(\rho)\|_1.
\end{align}

In Section \ref{sec:nonzero_alpha} we examine
the impact of imperfect $Z$ rotations on the gate error contribution
to additive errors. In Section \ref{sec:DE}, we analyze the effect of depolarizing
errors on the gate error contribution
to additive errors. Then in Section \ref{sub:spam}
we look at state preparation and measurement errors and their contributions
to additive errors.

\subsection{Errors in $Z$ Rotations}\label{sec:nonzero_alpha}

In section \ref{sub:theta}, we described a unitary operation
$U$, which involved applying the rotation $Z_{\pi/2}(0)$. 
Suppose that we can't implement $Z_{\pi/2}(0)$, but 
instead can implement $Z_{\pi/2}(\alpha)$.
Let $U'$ be the gate that results when $Z_{\pi/2}(0)$ is
replaced by $Z_{\pi/2}(\alpha)$ in
Eq. \eq{U_seq}.
Let $\sop U$ and $\sop U'$ label the corresponding
CPTP maps.

Using a similar triangle inequality as in Eq. \eq{triangle_errors},
we have that 
\begin{align}
|\tr(M_i&\left(\sop U^k-(\sop U')^k\right)(\rho)|\nonumber\\
&\leq 2k\max_{\substack{\rho\in Pos(2)\\\tr(\rho)=1}} 
\left\|\left(\sop Z_{\pi/2}(0)-\sop Z_{\pi/2}(\alpha)\right)(\rho)\right\|_1\nonumber\\
&\leq 4k\left|\sin\left(\frac{\pi\alpha}{4}\right)\right|,
\end{align}
so a non-zero $\alpha$ contributes at most an amount $k\pi|\alpha|$
to $\delta_{\sop U,\sop U'}.$ For the additive error to be bounded, we require
$|\alpha|=O(1/k)$.

In Section \ref{sub:alpha}, we showed that using $O(N)$ applications
of $Z_{\pi/2}(\alpha)$, we could estimate $\alpha$ with standard deviation
$O(1/N).$ Assuming that the control of $\alpha$ is precise enough
to correct $\alpha$ to within the uncertainty of this estimate,
we can obtain a new $Z$ rotation $Z_{\pi/2}(\alpha')$ with $|\alpha'|=O(1/N).$
This improved rotation can them be used to implement the protocol
for estimating $\theta$ in Section \ref{sub:theta} with standard 
deviation $O(1/N)$. Notice that
both procedures ($\alpha$ and $\theta$ estimation) together use $O(N)$
applications of gates, so in the end, we can obtain an estimate
of $\theta$ the scales at the Heisenberg limit. 

In practice, it is unrealistic to assume that experimentalists have
arbitrarily precise controls, and so at some point, even if 
$\alpha$ is estimated very precisely, it can not be corrected. 
However, in that case, there is no need to obtain such a precise
estimate, for the very reason that it can not be corrected.

We note that the strategy employed in this section is very general,
and can be employed for general
CPTP errors. However, when the errors have certain structure, we
can do better, as in the case of depolarizing errors, which we
analyze in the next section. 

\subsection{Depolarizing Errors}\label{sec:DE}

We now consider the effect of depolarizing noise.  We look at the case
that each applied gate is accompanied by depolarizing noise
$\Lambda_\gamma$, where \begin{align}
\Lambda_\gamma(\rho)=\gamma\rho+(1-\gamma)\I/2. 
\end{align} 

If we have an experiment that involves a sequence of $k$
gates, and the probability of a certain outcome assuming no depolarizing noise
is $1/2+r$ (for $|r|\leq 1/2$), then in the presence of depolarizing noise,
the probability of that outcome will be $1/2+\gamma^k r$. This
gives a gate error of 
\begin{align}
\delta_{\Lambda_\gamma}=|r|(1-\gamma^k)\leq(1-\gamma^k)/2.
 \end{align} 

For depolarizing errors with $\gamma=.99$, which is reasonable for
many quantum systems, one could go to sequences of over 100 operations
 before the depolarizing error would overwhelm the $1/\sqrt{8}$ bound
 of Theorem \ref{main_theorem}. Thus if the depolarizing error is small
 compared to the uncertainty in state preparation and measurement error, Theorem \ref{main_theorem} says
 that our procedure will give more accurate estimates of the parameters
 of interest than could be obtained using standard procedures.

 In fact, in the case of depolarizing errors, because of their simple form,
  one can do better than simply incorporating
 them into additive errors. The procedure of Section \ref{sec:phase_estimation}
 can be re-analyzed in the presence of depolarizing errors, allowing for more
 precise bounds. In
 the interest of conciseness and clarity, we relegate this analysis to later work.

\subsection{State Preparation Errors and Measurement Errors}\label{sub:spam}

State preparation and measurement errors (SPAM) are handled very well in general
by our procedure. This is because SPAM errors contribute a constant 
additive error ($\delta_{M_i,M_i'}+\delta_{\rho,\rho'}$) no matter what gates or operations 
are applied in between state preparation and measurement. As long as these 
additive errors are
not too large, our protocol works. However, there is a challenge in bounding
state preparation errors. Up until this point, we have tried to make as
few assumptions as possible. However, without good gates or good measurements,
it is very difficult to empirically bound the fiducial state preparation error.
Therefore, we do
have to make an assumption: we assume the the experimenter has an upper bound on the trace distance
between their true state preparation $\rho_{\ketbra{0}{0}}$
and the ideal state preparation $\ketbra{0}{0}$. (Once gates have been roughly calibrated,
 better bounds on this distance can then be obtained.) In many experimental set-ups,
the prepared state will be extremely close to the ideal \cite{MK13,JMS+12,Robledo2011}. We have
\begin{align}
\delta_{\ketbra{0}{0},\rho_{\ketbra{0}{0}}}\geq\frac{1}{2}\|\rho_{\ketbra{0}{0}}-\ketbra{0}{0}\|_1.
\end{align}

Now given the initial state $\rho_{\ketbra{0}{0}}$ and our
faulty gates $Z_{\pi/2}(\alpha)$ and $X_{\pi/4}(\epsilon,\theta)$,
we would like to create states that are close in trace distance
to $\ket{+}$ and $\ket{\rightarrow}$.

We will use the states
\begin{align}
\rho_{\ketbra{+}{+}}&=\sop Z_{\pi/2}(\alpha)\sop X_{\pi/4}(\epsilon,\theta)^2(\rho_{\ketbra{0}{0}}),\nonumber\\
\rho_{\ketbra{\rightarrow}{\rightarrow}}&=\sop X_{\pi/4}(\epsilon,\theta)^{6}(\rho_{\ketbra{0}{0}}).
\end{align}

Let $\xi_1=\max\{\epsilon,\theta,\alpha\}$ and $\xi_2=\max\{\epsilon,\theta\}$. Then using
the triangle inequality, one
can calculate that
\begin{align}
\frac{1}{2}\|\rho_{\ketbra{+}{+}}-\ketbra{+}{+}\|_1\leq&
\frac{\xi_1}{2}\left(\frac{\pi^4}{8}\left(12+4\pi+\pi^2\right)\right)^{1/4}\nonumber\\
&+\delta_{\ketbra{0}{0},\rho_{\ketbra{0}{0}}}+O(\xi_1^{5/4}),\nonumber\\
\frac{1}{2}\|\rho_{\ketbra{\rightarrow}{\rightarrow}}-\ketbra{\rightarrow}{\rightarrow}\|_1\leq&
\frac{1}{2}\left(\frac{9\pi^2}{2}\theta^2\epsilon^2\right)^{1/4}\nonumber\\
&+\delta_{\ketbra{0}{0},\rho_{\ketbra{0}{0}}}+O(\xi_2^{5/4}).
\end{align}
In other words, we can create approximate state preparations, which
induce additive errors of the order of the size of the errors in the 
gates used to create them, plus the base additive error from
incorrect preparation of $\ketbra{0}{0}$.

Let $W$ be a measurement operator that is ideally
close to $\ketbra{0}{0}.$
In Appendix \ref{app:bounding} we show how to bound
\begin{align}\label{eq:def_delta_M}
\delta_{\ketbra{0}{0},W}=\max_\rho\left|\tr(W\rho)-
\tr(\ketbra{0}{0}\rho))\right|
\end{align}
given access to the state $\ketbra{0}{0}$ and any other state. 
As usual, if $\rho_{\ketbra{0}{0}}$ is used instead of
$\ketbra{0}{0},$ the difference in outcomes will be bounded
by $\delta_{\ketbra{0}{0},\rho_{\ketbra{0}{0}}}$.

A rotation similar to what is used
in state preparation can be applied to $W$ to obtain
$W_{\ketbra{+}{+}}$ (an operator close to $\ketbra{+}{+}$),
and the additive error for this measurement can be found using
the standard triangle inequality strategy we have employed multiple times.


\section{Non-adaptive Heisenberg limited phase estimation}\label{sec:phase_estimation}

In this section, we will prove Theorem \ref{main_theorem}. First, in
Section \ref{sec:WE}, to set up  the main ideas, we review, and
slightly improve, the proof of Heisenberg scaling without additive
errors  by Higgins et al. \cite{HBB+7}. This sufficiently motivates
our proof in Section \ref{sec:AE}.

\subsection{Heisenberg limit without errors}\label{sec:WE}

Our proof of Theorem \ref{main_theorem} is based on the non-adaptive
phase estimation procedure of Higgins et al. \cite{HBB+7}, which states
\begin{thm}\label{thm:Higgins}\cite{HBB+7}
Say that we can perform two families of experiments, $\ket{0}$-experiments
and $\ket{+}$-experiments, indexed by $k\in\mathbb{Z}$, whose
probabilities of success are, respectively,
\begin{align}
p_0(A,k)&=\frac{1+\cos(kA)}{2},\label{eq:p0}\\
p_+(A,k)&=\frac{1+\sin(kA)}{2}.\label{eq:pp}
\end{align}
Also assume that performing either of the $k^{\text{th}}$ experiments
takes time proportional to $k$. Then, an estimate $\hat{A}$ of $A\in(-\pi,\pi]$
with standard deviation $\sigma(\hat{A})$ can be obtained in time
$T=\mathcal{O}(1/\sigma(\hat{A}))$ using non-adaptive measurements.
\end{thm}

We reprove Theorem \ref{thm:Higgins} because we use new techniques
that give improved analytic bounds on the scaling of $T\sigma(\hat A)$
compared to \cite{HBB+7}.
These techniques might additionally be of broader use.

For a given $k$, let $\hat{a}_{0}$ ($\hat{a}_{+}$) be the   number
of successful outcomes of the $\ket{0}$- ($\ket{+}$-) experiments
respectively if $M$ samples are taken of each experiment.  Then one can
obtain an estimate $\widehat{kA}$ for $kA$ with
standard deviation $\sigma(\widehat{kA})$:
\begin{align}
\widehat{kA}&= \textrm{atan2}\left[\hat{a}_+-M/2,
\hat{a}_0-M/2\right]\in(-\pi,\pi],\\ \nonumber
\sigma(\widehat{kA})&\propto\frac{1}{\sqrt{M}}.
\end{align}
It is tempting to use this to get an estimate $\hat{A}=\widehat{kA}/k$ for $A$,
apparently with standard deviation
\begin{align}
\sigma(\hat{A})&\propto\frac{1}{k\sqrt{M}}\propto\frac{1}{T},
\end{align}
which gives Heisenberg scaling if $M$ is independent of $k$.
Unfortunately, this estimate is deceptive as it
is only correct up to factors of $\frac{2n\pi}{k}$,
$n\in\mathbb{Z}$, due to the unknown principle range of $kA$.

To determine the correct range of $\widehat{kA}/k$ while still retaining
Heisenberg scaling, Higgins et al. instead  sample distributions with
a range of values of $k$. In particular,  they choose $k$ from $\{k_1,\dots,k_K\}$,
with $k_j=2^{j-1}$. Let $\hat{A}_{j}=\widehat{k_jA}/k_j$ be an estimate
of $A$ obtained from setting $k=k_j$. Then $\hat{A}_1$
 is used to restrict estimates $\hat{A}_j$ for $j>1$ to the range
$(\hat{A}_1-\pi/2,\hat{A}_1+\pi/2].$  Continuing in this way, we
assume  $\hat{A}_{j+1}\in (\hat{A}_{j}-\pi/2^j,\hat{A}_{j}+\pi/2^j]$.
(This restriction differs slightly from Higgins et al.,
in  which they assume  $\hat{A}_{j+1}\in
(\hat{A}_{j}-\pi/3^j,\hat{A}_{j}+\pi/3^j]$.  This small difference
allows us to apply  much stronger bounds to the probability of failure
at any step.)

We immediately see that $\hat{A}_{K}$ will only be in the correct
principle range conditional on all prior estimates $\hat{A}_j$ being
within $\pm\frac{\pi}{2 k_j}$ of the actual value of $A$. In other
words, the probability
\begin{align}
\label{eq:p_error}
p_{\text{error}}(k_jA)\equiv  P\left[k_j(\hat{A}_j-A)\ge \frac{\pi}{2}
\bigvee k_j(\hat{A}_j-A)< -\frac{\pi}{2}\right]
\end{align}
 must be small for all $j$, where the average is taken over
 possible estimates $k_j\hat{A}_j$. (We define $p_{\text{error}}(k_jA)$
 as stated instead of as $P\left[|k_j(\hat{A}_j-A)|\ge \frac{\pi}{2}\right]$
 in order to obtain slightly better bounds.) Any one such error
occurring will lead to an incorrect range of $\hat{A}_{K}$
and thus an incorrect estimate of $\hat{A}$. As the precise value of
$p_{\text{error}}$ has a significant impact in evaluating the scaling
constant of $\sigma(\hat A)=\mathcal{O}(\frac{1}{T})$, a careful bound on
$p_{\text{error}}$ is required. In Lemma \ref{lemm:perror_bound} in
Appendix~\ref{sec:binomial}, we show that if $M_j$ samples
are taken of each of the $k_j\tth$ $\ket{0}$- and $\ket{+}$-experiments,
\begin{align} \label{eq:pmax}
p_{\text{max}}(M_j)\equiv \frac{1}{\sqrt{2\pi M_j}2^{M_j}}
>p_{\text{error}}(k_jA).
\end{align}
This is a stronger bound than what appears in Higgins et al.,
which is derived from Hoeffding's bound. This stronger bound
in turn allows us to obtain a better analytic bound on the variance
of our final estimate.

To calculate the variance of our estimate, we note that if no
errors occur in our principal range estimates for all $k_j<k_h$,
then the maximum error in our estimate is
\begin{align}
\xi(h)=\frac{2\pi}{2^h}.
\end{align}
Furthermore, even if we have no errors in our principal range
estimates, our final estimate can still differ from the
true value by at most
\begin{align}
\overline\xi(K)=\frac{2\pi}{2^{K+1}}.
\end{align}

Thus, we can bound the variance of our estimate $\hat{A}$ of $A$ with
\begin{align}
\label{eq:variance_upper_bound}
\sigma^2(\hat{A})\leq&\big(1-p_{\text{error}}(k_KA)\big)\overline\xi(K)^2\nonumber\\
&+
\sum_{j=1}^K\xi(j)^2p_{\text{error}}(k_jA)\prod_{i=1}^{j-1}(1-p_{\text{error}}(k_iA))\nonumber\\
\leq&\left(1-p_{\text{max}}(M_K)\right)\overline\xi(K)^2+
\sum_{j=1}^K \xi(j)^2p_{\text{max}}(M_j).
\end{align}
Note that the first term is a variance contribution from the event of
\emph{no} errors whereas the second term is the contribution in the
event where errors arise.

We assume that running the $k_j^\tth$ $\ket{0}$- or $\ket{+}$-experiment
takes time $k_j$. Then
the total time required for our estimate is
\begin{align}
T=2\sum_{j=1}^K2^{j-1}M_j.
\end{align}
As in \cite{HBB+7}, setting $\delta_{M_j}(\sigma^2(\hat{A})T^2)=0$, we
find Heisenberg scaling can be attained by setting
\begin{align}\label{eq:Mjsolve}
M_j=\alpha(K- j)+\beta
\end{align}
for $\alpha,\beta\in\mathbb{Z}^+$. The sum in Eq.~\eq{variance_upper_bound} can be performed by making the replacement $p_{\text{max}}(M_j)\le\frac{1}{\sqrt{2\pi \beta}2^{M_j}}$. One finds that $\alpha > 2$ is necessary to prevent the sum from growing faster than $\sim 4^{-K}$, which results in poorer-than-Heisenberg scaling. We obtain
\begin{align}
\label{eq:Heisenberg_constant}
\sigma^{2}(\hat{A})&\le \frac{\pi^2}{4^K}
\left[1+p_{\textrm{max}}(\beta)\left(3+\frac{16}{2^\alpha-4}\right)\right],\nonumber \\ \nonumber
T&<2^{K+1}(\alpha+\beta), \\
\sigma(\hat{A})T&\le 2\pi(\alpha+\beta)
\sqrt{1+p_{\textrm{max}}(\beta)\left(3+\frac{16}{2^\alpha-4}\right)},
\end{align}
which holds for all $K>0$.

Thus Heisenberg scaling can be obtained for any $\alpha > 2,$ $ \beta > 0$. Optimizing Eq.~\eq{Heisenberg_constant} 
over the integers gives $\sigma(\hat{A})T<12.4\pi$ at $\alpha=3$, $\beta=1$. Better bounds of
$\sigma(\hat{A})T<10.7\pi$ can be attained at $\alpha=5/2$, $\beta=1/2$, where
fractional values of $M_j$ means one rounds up to the nearest integer value and performs that many experiments.
This improved bound also uses a more sophisticated analysis of Eq. \eq{variance_upper_bound},
in which we pull out the last $j=K,K-1,...,K-z$ terms from the sum in
Eq.~\eq{variance_upper_bound}, and use
$p_{\text{max}}(M_j)\le\frac{1}{\sqrt{2\pi M_{K-z}}2^{M_j}}, $ for values
of $j < K-z$
to transform the remainder into a geometric sum. These analytic bounds
are significant practical improvements over those in \cite{HBB+7}
where $\sigma(\hat{A})T< 54\pi$ at $\alpha=8\ln{2}$, $\beta=23/2$.

We compare our result to the scaling of various other phase estimation procedures
 (including maximum likelihood and procedures using
entanglement) in Appendix \ref{sec:compare_procedures}.
While the improved analysis of this section gives us better analytic scaling
than was previously known for non-adaptive phase estimation, our main motivation
is to obtain better results in the presence of additive errors.
The new analysis allows us to include much larger additive errors than would
have been possible previously.

\subsection{Including additive errors}\label{sec:AE}

We now consider the case that the success probabilities of
our experiments differ from the ideal probabilities by
 constant factors $\delta_0(k_j)$ and $\delta_+(k_j)$ as
\begin{align}
p_0(A,k_j)=\frac{1+\cos{k_jA}}{2}+\delta_0(k_j)\\
p_+(A,k_j)=\frac{1+\sin{k_jA}}{2}+\delta_+(k_j).
\end{align}
Let
\begin{align}
\delta_j=\max\{|\delta_0(k_j)|, |\delta_+(k_j)|\}.
\end{align}
Suppose we use exactly the same procedure to estimate $A$ as in
the case of no additive errors. Then
in Lemma \ref{lemm:perror_delta} in Appendix \ref{sec:binomial}
we show that now,
\begin{align}
p_{\textrm{max}}(M_j,\delta_j)\equiv&\frac{1}{\sqrt{2\pi}(1-\sqrt{8}\delta_j)}
\frac{\left(1-\frac12(1-\sqrt{8}\delta_j)^2\right)^{M_j}}{\sqrt{M_j}}\nonumber\\
>&p_{\text{error}}(k_jA),
\end{align}
where $p_{\text{error}}(k_jA)$ is defined in Eq. \eq{p_error}.

 Now consider replacing $M_j$ by $F(\delta_j,M_j)\times M_j$, where
 $F(\delta_j,M_j)$ is
 \begin{align}
F(\delta_j,M_j)=\frac{\log\left(\frac12(1-\sqrt{8}\delta_j)^{1/M_j}\right)}
 {\log\left(1-\frac12(1-\sqrt{8}\delta_j)^2\right)}.
 \end{align}
Then as long as $\delta_j<1/\sqrt{8}\approx 0.354$,
we have
\begin{align}
p_{\text{max}}(F(\delta_j,M_j)M_j,\delta_j)
\leq&\frac{1}{\sqrt{2\pi M_j}2^{M_j}}.
\end{align}
This bound is the same as Eq. \eq{pmax}.
This means that by increasing the number of samples of the
$j\tth$ experiment by a factor
$F(\delta_j,M_j)$, we can get the same error bounds as if there
were no additive errors $\delta_0(k_j)$ and $\delta_+(k_j)$.

Suppose there is some smallest $h$ such that $\delta_h\ge1/\sqrt{8}$.
In this case,
no matter how many times we repeat the experiments,
no matter how many samples we take, $p_{\text{error}}(h)$ will not be bounded.
However, we can still use the procedure of the previous
section to obtain an estimate of $\hat{A}$ with variance
proportional to $4^{-(h-1)}$, by using
$F(\delta_j,M_j)M_j$ samples for each $j\leq h-1$, proving the second part of
Theorem \ref{main_theorem}.

Furthermore, if
\begin{align}
\sup_j\delta_j=1/\sqrt{8}, \textrm{ but } \max_j\delta_j\neq1/\sqrt{8},
\end{align}
then we can always increase the number
of samples taken of each experiment in order to counteract the effect of additive errors.
This means that we can obtain arbitrarily accurate estimates.
However, the size of the required $F(\delta_j,M_j)$ blows up,
so we will no longer have Heisenberg scaling.

However, if $\sup_j\delta_j<1/\sqrt{8}$, then for all $j$, we have
\begin{align}
\delta_j<1/\sqrt{8}-e\equiv \delta'
\end{align}
for some  constant $e$. Then if we take $F_jM_j$ samples
of the $j\tth$ iteration, where
 \begin{align}
 F_j=\frac{\log\left(\frac12(1-\sqrt{8}\delta')^{1/M_j}\right)}
 {\log\left(1-\frac12(1-\sqrt{8}\delta')^2\right)},
 \end{align}
 we can attain the correct bounds on $p_{\textrm{error}}.$
If we set $M_j=\alpha(K- j)+\beta$ as before,
$M_j$ is a monotonically decreasing sequence in $j$, so
$F_j$ is a monotonically increasing sequence. Thus, we have $F_j\le
F_K$ for all $j=1,2,\dots,K$.

If for each $M_j$ we replace $M_j$ by $F_jM_j$, we have increased the
total time required by the procedure by at most a constant factor
$F_K$, and obtained at least as good a $p_{\textrm{error}}$ at each step as in
the case without any errors $\delta_0(k_j)$ or $\delta_+(k_j)$. Thus
we can obtain Heisenberg scaling, where $T\sigma_A$ increases by the constant
$F_K$ compared to the case without additive errors $\delta_0(k_j)$ or
$\delta_+(k_j)$. This completes the proof of Theorem \ref{main_theorem}.

\section{Conclusions and Open Problems}

There are many ways to extend and refine the ideas
of this paper. In particular, while the techniques
described here seem to apply broadly for single-qubit
operations, it would be both interesting theoretically and of great
practical use if these procedures could be extended to
multi-qubit operations.

Additionally, there is much room for improvement in terms
of error analysis. In this work, we've suggested treating
depolarizing or amplitude damping noise as contributing
to additive errors. However, this is essentially a worst-case
scenario, in which every process adversarially drives
you away from the desired state by as much as possible.
In reality,  we would expect the repeated applications of
the gate to have a twirling effect, thus mitigating, or at least averaging,
the  effect of noise, as in randomized benchmarking \cite{KLR+08}.
In addition it would be of practical relevance to analyze the
case where $\theta_A$ and $\epsilon_A$ are not fixed, but
shift over time.

Finally, at least on the surface, our procedure has many similarities
to randomized benchmarking:
both procedures are (more or less) robust to SPAM errors, and involve
applying increasingly lengthy sequences of operations. These
similarities draw the question: is there an explicit
connection between phase estimation and randomized benchmarking?

\section{Acknowledgments}
We are grateful for helpful discussions with Colm Ryan, Blake Johnson,
and Marcus da Silva. SK acknowledges support from the Department of Defense.
GHL acknowledges support from the ARO Quantum Algorithms
program.
TJY acknowledges the support of the NDSEG fellowship program.

\bibliographystyle{plain}
\bibliography{Flipflop}

\appendix

\section{Bounds on $p_{\textrm{error}}$}\label{sec:binomial}

In this section, we bound the probability of making
an error at any step during our estimation procedure.
An error occurs at the $j^\tth$ iteration if
\begin{align}\label{eq:errorCond}
\left(k_j(\hat{A}_j-A)\ge \frac{\pi}{2}\right)\bigvee \left(k_j(\hat{A}_j-A)< -\frac{\pi}{2}\right).
\end{align}
In the below analysis, we replace $k_j\hat{A}$ with the variable $\hat{\varphi}$
and $k_jA$ with $\varphi$. In Lemma \ref{lemm:perror_bound}, we consider
the case without additive errors $\delta_0(k_j)$ and $\delta_+(k_j)$,
and in Lemma \ref{lemm:perror_delta} we include these errors.

\begin{figure*}
\subfloat[\label{fig:aaplot}]{%
  \includegraphics[height=8cm,width=.49\linewidth]{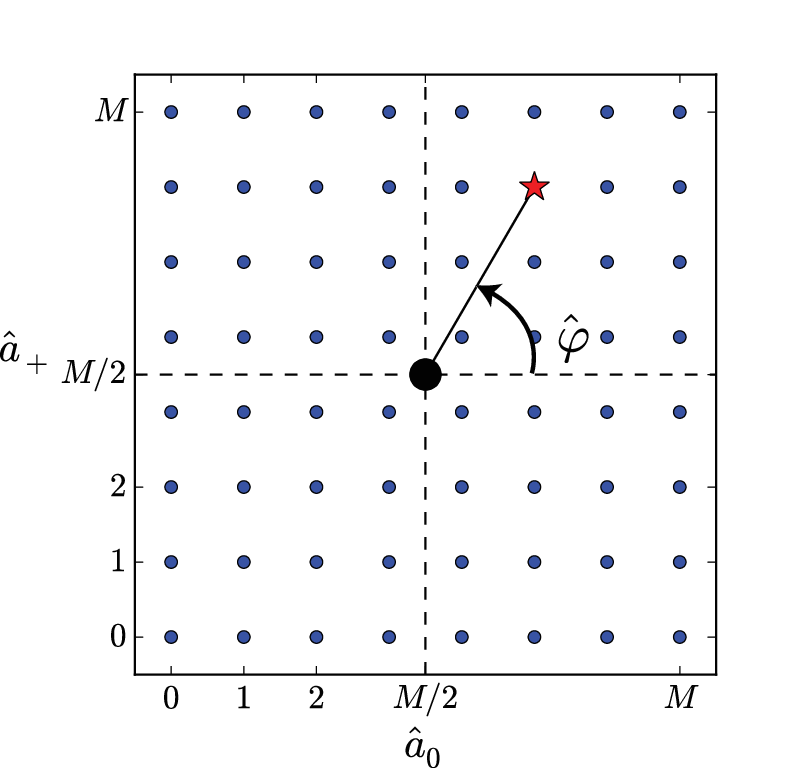}%
}\hfill
\subfloat[\label{fig:aaplot2}]{%
  \includegraphics[height=8cm,width=.49\linewidth]{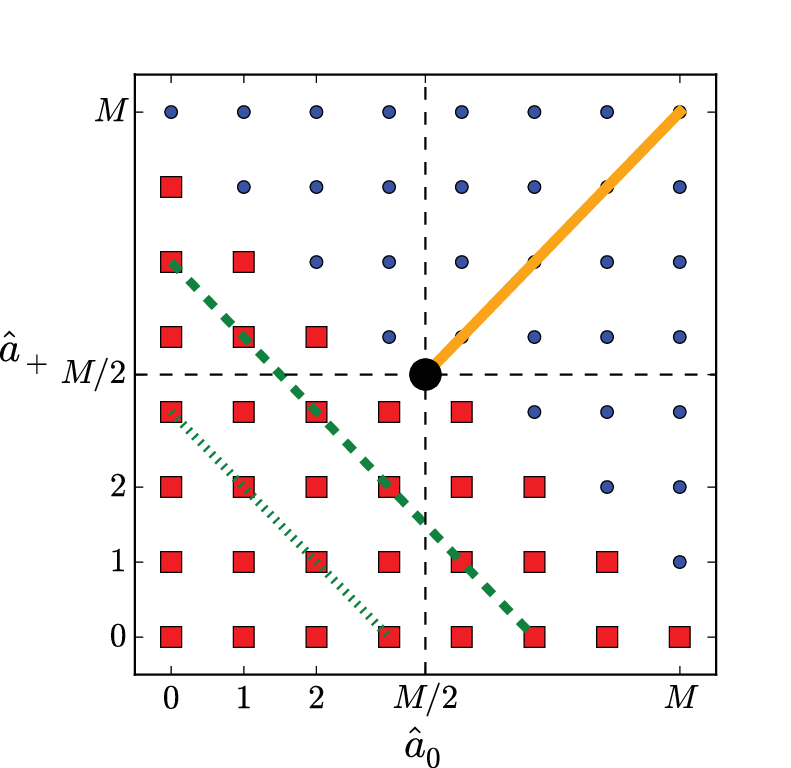}%
}
\caption{In Fig. \ref{fig:aaplot} we show how to calculate
$\hat\varphi$
given $\hat{a}_0$ and $\hat{a}_+$. Note
$\hat{a}_0$ and $\hat{a}_+$ can take values in $\{0,1,\dots,M\}$, so
the blue circular dots represent the possible outcomes $(\hat{a}_0,\hat{a}_+)$.
In Fig. \ref{fig:aaplot2}, we consider the case that
$\varphi$ lies
along the orange line in the upper right quadrant, corresponding to the maximum value of $p_{\textrm{error}}$.
In this case, all of the points
with red square markers correspond to errors. We sum the
probability of being at one of these points by first
calculating the probability of being at one of the points
intersected by the green dashed lines.}
\label{}
\end{figure*}

\begin{lem}\label{lemm:perror_bound}
For $\varphi\in(-\pi,\pi]$ let
\begin{align}
p_0&=\frac{1+\cos(\varphi)}{2},\label{eq:p0theta}\\
p_+&=\frac{1+\sin(\varphi)}{2}\label{eq:pptheta},
\end{align}
and let $\hat{a}_{0}$ (respectively $\hat{a}_{+}$) be drawn from the binomial distribution
$\sop B(M, p_0)$
(resp. $\sop B(M, p_+)$).
Let
\begin{align}
\hat{\varphi}&= \emph{atan2}\left(\frac{2}{M}\hat{a}_{+}-1,\frac{2}{M}\hat{a}_{0}-1\right)
\end{align}
be an estimate for $\varphi$ (and if $a_0=a_+=M/2$, then $\hat{\varphi}$
is chosen uniformly at random from $(-\pi,\pi]$). Then
\begin{align}
p_{\emph{error}}(\varphi)<\frac{1}{\sqrt{2\pi M}2^{M}}
\end{align}
where
\begin{align}\label{eq:perro}
p_{\emph{error}}(\varphi)\equiv P\left[(\hat{\varphi}-\varphi\ge \pi/2)\bigvee(\hat{\varphi}-\varphi<-\pi/2)\right],
\end{align}
and the probability is taken over the possible outcomes $\hat{a}_0$ and $\hat{a}_+.$
\end{lem}

\begin{proof}

While Hoeffding's inequality gives a loose bound on $p_{\textrm{error}}(\varphi)$, we will
use a geometric interpretation to obtain a stronger and asymptotically exact result.
In particular, we can extract an estimate $\hat{\varphi}$ for
$\varphi$ graphically by plotting the value of $\hat{a}_0$ and
$\hat{a}_+$ on orthogonal axes, as shown in  Figure \ref{fig:aaplot}.

Before we take advantage of this geometric interpretation,
we first will show that $p_{\textrm{error}}(\varphi)$ is largest when
$\varphi=\pi/4$, and thus we need only analyze
$p_{\textrm{error}}(\pi/4)$.

We introduce the substitution $\hat{y}=\frac{2}{M}\hat{a}_{+}-1$,
$\hat{x}=\frac{2}{M}\hat{a}_{0}-1$ and
consider the inner product
\begin{align}
\label{eq:r_length}
\hat{r}=(\hat{x},\hat{y}).(\cos\varphi,\sin\varphi).
\end{align}
Note that $p_{\textrm{error}}(\varphi)$ corresponds to the probability
that $\hat{r}$ is less than $0$ (with some small
correction because of one sided error).
In the limit of very large $M,$ $\hat{r}$
becomes a weighted sum of two independent normal distributions,
and is hence a normal distribution itself. As
normal distributions are completely characterized by
their mean and variance, in this limit, $p_{\textrm{error}}(\varphi)$ depends only on the mean
and variance of $\hat{r}$. In  particular, $p_{\textrm{error}}(\varphi)$ will
be largest when the mean of this distribution is smallest and the variance is largest.

Using the well-known properties of binomial
distributions and properties of sums of independent distributions,
we have
\begin{align}
\text{E}\left[\hat r\right]&= 1,&
\text{Var}\left[\hat r\right]&= \frac{1}{2M}\sin^2(2\varphi).
\end{align}
Thus the variance of $\hat{r}$ and hence the probability of error is largest when $\varphi=\pi/4+q\pi/2$
for any integer $q$. When $M$ is not large, we verify (see Figure \ref{fig:prob_varphi_plot})
that $p_{\textrm{error}}(\varphi)$ is indeed largest at $\varphi=\pi/4$.

This leads to a drastic simplification --- we need only bound $p_{\textrm{error}}(\pi/4)$.
($p_{\textrm{error}}(\varphi)$ for $\varphi=\pi/4+q\pi/2$ is the same as $\varphi=\pi/4$ by symmetry.)
 This corresponds to
$\varphi$ lying along the orange line in Figure \ref{fig:aaplot2}. Then
an error occurs when values of $\hat{a}_0$ and $\hat{a}_+$
correspond to the red square markers on Figure \ref{fig:aaplot2}.
Thus to bound $p_{\textrm{error}}(\pi/4)$, we calculate the
probability of ending up at any one of the red markers. We do this by
summing over the cases where $(\hat{a}_0+\hat{a}_+)$ is constant and
no greater than $M$, corresponding to the dashed green lines in Figure
\ref{fig:aaplot2}.

\begin{figure}
\includegraphics[width=\columnwidth]{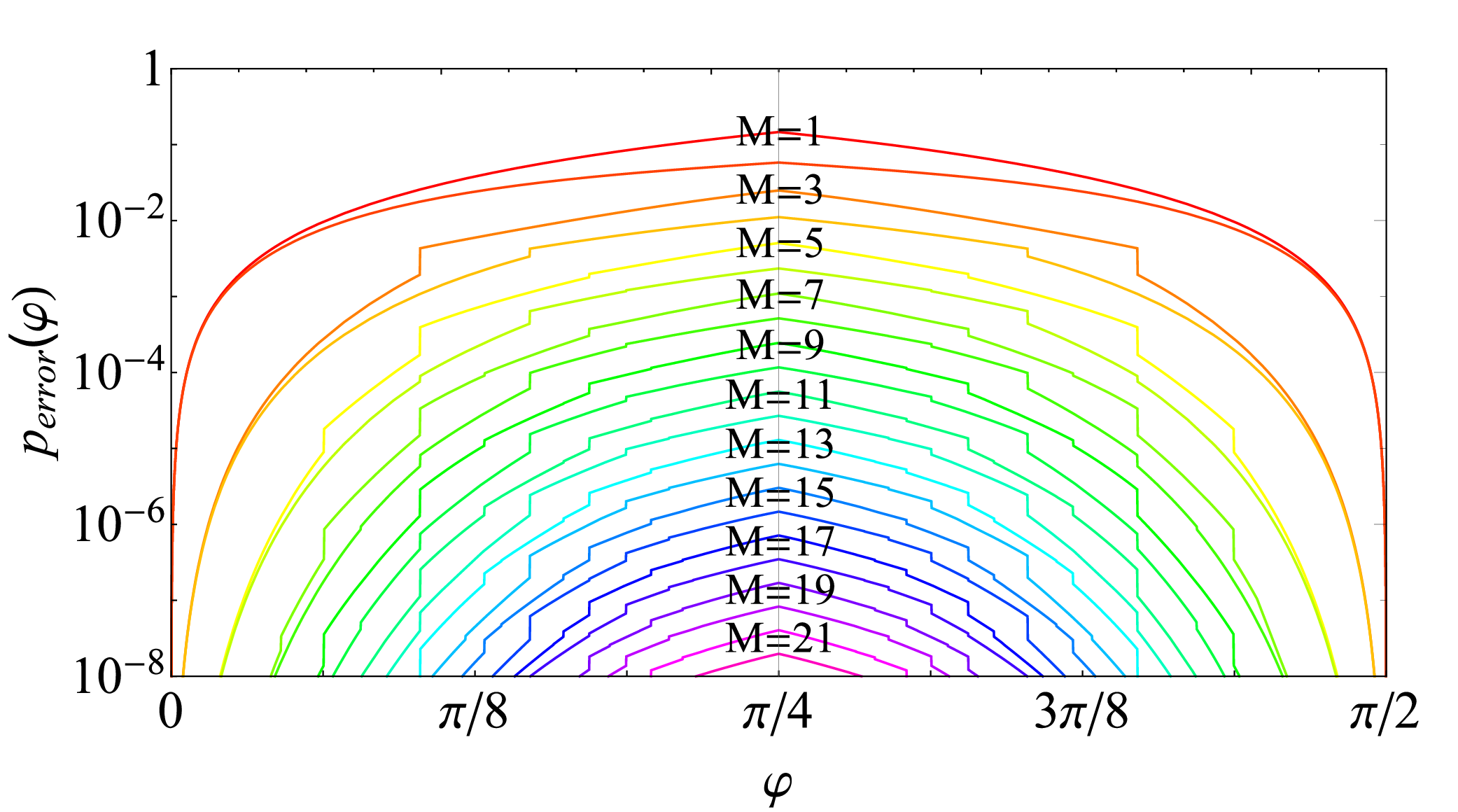}
\caption{\label{fig:prob_varphi_plot}
Exact probability of error as a function of $\varphi$ by enumeration
over all possible outcomes $\hat{a}_0$ and $\hat{a}_+$ that lead to
errors in $\hat{\varphi}$ defined in Eq.~\eq{perro}. Different
lines correspond to the labeled number of repeats $M=1,2,...$ from
the top. Observe that the maximum occurs at $\varphi = \pi/4$ for all M.}
\end{figure}

For $\varphi=\pi/4,$ we have $p_0=p_+\equiv p=(2+\sqrt2)/4$,
and the probability of finding $\hat a_0=a_0$ and $\hat a_+=a_+$ is
\begin{align}
\label{eq:probability_joint}
P\left[a_0,a_+\right]=&\binom{M}{a_0}\binom{M}{a_+}\left(\frac{p}{1-p}\right)^{a_0+a_+}\left(1-p\right)^{2M}.
\end{align}
The probability of lying on a line $\hat a_0+\hat a_+=b$ is
\begin{align}
P_{\textrm{diag}}(b)&= \sum_{a_0=0}^MP\left[a_0,b-a_0\right]\\\nonumber
&=\binom{2M}{b}\left(\frac{p}{1-p}\right)^b\left(1-p\right)^{2M}.
\end{align}
Summing over the lines of constant $(a_0+a_+)$ up to $M-1$
and including half of the line $(a_0+a_+)=M$, we have,
\begin{align}
p_{\text{error}}(\pi/4)&=\sum_{b=0}^{M}P_{\textrm{diag}}(b)-\frac{1}{2}P_{\textrm{diag}}(M)\\\nonumber
&=(p(1-p))^M\frac{(2M)!}{(M!)^2}\left(H\left(M,\frac{1-p}{p}\right)-\frac{1}{2}\right)\\\nonumber
&=\frac{(2M)!}{8^M(M!)^2}\left(H\left(M,\frac{2-\sqrt2}{2+\sqrt{2}}\right)-\frac{1}{2}\right),
\end{align}
where
\begin{align}
H(M,z)=\sum_{x=0}^M\frac{(M!)^2z^x}{(M-x)!(M+x)!}.
\end{align}
As the $x=0$ term is $1$ and the ratio of successive terms in
$H(M,z)$ is
\begin{align}
\frac{M-x}{1+M+x}z< z,
\end{align} we can bound this sum with a
geometric series:
\begin{align}
H(M,z) <\sum_{x=0}^M z^x<\frac{1}{1-z}.
\end{align}
Using Stirling's approximation $n! \sim \sqrt{2\pi n}(n/e)^n$ and
noting that the fractional error of the approximation decreases
monotonically with $n$, we obtain the 
remarkably simple bound
\begin{align}
\label{eq:error_no_additive}
p_{\text{error}}(\varphi)<\frac{1}{\sqrt{2\pi M}2^M},
\end{align}
which is tight in the limit $M\rightarrow\infty$.

\end{proof}

We now include additive errors into the analysis:
\begin{lem}\label{lemm:perror_delta}
For $\varphi\in(-\pi,\pi]$ and $\delta_0,$ $\delta_+$ such that $|\delta_0|,|\delta_+|\leq \delta<1/\sqrt{8}$, let
\begin{align}
p_0&=\frac{1+\cos(\varphi)}{2}+\delta_0,\\
p_+&=\frac{1+\sin(\varphi)}{2}+\delta_+,
\end{align}
and let $\hat{a}_{0}$ $(\hat{a}_{+})$ be drawn from the binomial distribution
$\sop B(M, p_0)$
$(\sop B(M, p_+))$.
Let
\begin{align}
\hat{\varphi}&= \emph{atan2}\left(\frac{2}{M}\hat{a}_{+}-1,\frac{2}{M}\hat{a}_{0}-1\right)
\end{align}
be an estimate for $\varphi$ (and if $a_0=a_+=M/2$, then $\hat{\varphi}$
is chosen uniformly at random from $(-\pi,\pi]$). Then
\begin{align}
p_{\emph{error}}(\varphi,\delta_+,\delta_-)<\frac{e}{2\pi}\frac{1}{1-\sqrt{8}\delta}\frac{\left(1-\frac12\left(1-\sqrt8\delta\right)^2\right)^M}{\sqrt{M}}
\end{align}
where
\begin{align}\label{eq:perror2}
p_{\emph{error}}(\varphi,\delta_+,\delta_-)\equiv P\left[(\hat{\varphi}-\varphi\ge \pi/2)\bigvee(\hat{\varphi}-\varphi<-\pi/2)\right],
\end{align}
and the probability is taken over the possible outcomes $\hat{a}_0$ and $\hat{a}_+.$
\end{lem}

\begin{proof}
This proof will be similar to the proof of Lemma \ref{lemm:perror_bound},
so we will omit some of the details if they parallel the previous result.
 As done in
Lemma \ref{lemm:perror_bound}, we introduce the substitution
$\hat{x}=\frac{2}{M}\hat{a}_0-1$, $\hat{y}=\frac{2}{M}\hat{a}_+-1$
and consider
\begin{align}
\hat{r}=(\hat{x},\hat{y}).(\cos\varphi,\sin\varphi).
\end{align}
We find in this case that in the limit of large $M,$
\begin{align}
\text{E}\left[\hat r\right]&= 1+2(\delta_0 \cos{\varphi}+\delta_+ \sin{\varphi}),\\ \nonumber
\text{Var}\left[\hat r\right]&=\frac{1-\left(\cos{\varphi}(2\delta_0+\cos{\varphi})\right)^2-
\left(\sin{\varphi}(2\delta_++\sin{\varphi})\right)^2}{M}.
\end{align}
As explained in the proof of Lemma \ref{lemm:perror_bound},
 $p_{\textrm{error}}$ is maximized when we simultaneously minimize $\hat{r}$'s expectation and maximize its
variance. Using $|\delta_0|,|\delta_+|\le \delta$, we have
\begin{align}
\text{E}\left[\hat r\right]&\ge 1+\sqrt{8}\delta\cos{(\varphi-s)},\\\nonumber
\text{Var}\left[\hat r\right]&\le\frac{1}{M}\left(1-\cos{\varphi}^2
\min{\left[1,(2\delta+\sqrt{2}\cos{s}\cos{\varphi})^2\right]}\right.\\\nonumber
&\quad\left.-\sin{\varphi}^2\min{\left[1,(2\delta+\sqrt{2}\sin{s}\sin{\varphi})^2\right]}\right).
\end{align}
where $s=\pi\left(\frac{1}{4}+\frac{j}{2}\right),\,j=0,1,2,3$ is used
to represent the signs of $\delta_0$ and $\delta_+$.
Thus, the worst-case bounds
\begin{align}
\text{E}\left[\hat r\right]&\ge 1-\sqrt{8}\delta,\no
\text{Var}\left[\hat r\right]&\le \frac{1}{M}\left(1-(\frac{1}{\sqrt{2}}-2\delta)^2\right),
\end{align}
are obtained when $\delta_0=\delta_+=-\delta$ (corresponding to $s=\pi+\pi/4$) and
$\varphi=\pi/4$, leading to $p_0=p_+\equiv p=(2+\sqrt{2})/4-\delta$. We thus
have $p_{\textrm{error}}(\varphi,\delta_+,\delta_-)\leq p_{\textrm{error}}(\pi/4,-\delta,-\delta)$.

The bound on $p_{\textrm{error}}(\pi/4,-\delta,-\delta)$ is then obtained by a calculation identical
to the proof of Lemma~\ref{lemm:perror_bound} from Eq.~\eq{probability_joint}
onwards, except with $p=(2+\sqrt{2})/4-\delta$. We obtain
\begin{align}
p_{\textrm{error}}(\pi/4,-\delta,-\delta)&=\frac{(2M)!}{4^M(M!)^2}\left(1-\frac12\left(1-\sqrt8\delta\right)^2\right)^M\no
&\quad\times \left(H\left(M,\frac{2-\sqrt{2}+4\delta}{2+\sqrt{2}-4\delta}\right)-\frac{1}{2}\right)\no
&< \frac{1}{\sqrt{2\pi}}\frac{1}{1-\sqrt{8}\delta}\frac{\left(1-\frac12\left(1-\sqrt8\delta\right)^2\right)^M}{\sqrt{M}}.
\end{align}
Observe that Eq.~\eq{error_no_additive} is recovered in
the absence of additive errors (i.e. when $\delta=0$).
\end{proof}

\section{Scaling of Phase Estimation Procedures}\label{sec:compare_procedures}

In Section \ref{sec:WE}, we gave an analytic bound on
the scaling of our Heisenberg-limited phase estimation
technique. Optimizing Eq.~\eq{variance_upper_bound} gave
$\sigma(\hat{A})T< 10.7\pi$.

This upper bound on the Heisenberg scaling constant should of course
be compared to lower bounds. A number of lower bounds are commonly cited in the literature, depending on the specification
of allowed resources. The best possible bound is $\sigma(\hat{A})T
\ge 1$~\cite{Bollinger1996}, often used in the atomic clocks community~\cite{Leibfried2004}. The resources required are similar to those used for our scheme, except that there is
no iteration from $j=1,...,K-1$, so only the largest $K$ experiment is
used. However, achieving this bound is only possible when the principle range
of $A$ is known -- a reasonable assumption when tracking well-characterized
frequencies, but not when $A$ is completely unknown.

The next largest bound on the scaling is $\sigma(\hat{A})T \ge \pi$~\cite{Berry2000},
which is achievable using quantum phase estimation. Unlike the above case, $A$ can be completely unknown initially.
However, this scheme requires the resource of entanglement between different
experimental runs with multi-qubit gates, or non-local measurements~\cite{Sanders1995}.
Such requirements are technically demanding, which motivates
entanglement-free schemes.

Reasonable lower bounds for the entanglement-free scenario can be
derived, but proving whether they are achievable remains an open
question. For each experiment at some $k_j$, (with $k_j$ as in Section \ref{sec:phase_estimation})
, the amount of information
we obtain about $A$ can be quantified by the Fisher information
\begin{align}
I(A,k_j)=\mathbf{E}\left[\left(\frac{d \log{p(A,k_j)}}{d A}\right)^2\right]=k_j^2,
\end{align}
where expectation over success and failure is taken. As the $2M_j$
repeats of the experiment are independent, the total information
obtained over all values of $k$ is $I=\sum_{j=1}^{K} I(A,k_j)2M_j$. In
the large $K$ limit, $I=\frac{2}{9}4^K(3\beta+\alpha)$. Using the
Cramer-Rao inequality~\cite{hodges1951} then bounds the variance of $\hat{A}$ obtained
via any unbiased estimator, such as maximum likelihood estimation, by
$\sigma^{2}(\hat A)\ge F^{-1}$. Thus we obtain
\begin{align}
\sigma(\hat{A})T \ge (\alpha+\beta)\sqrt{\frac{18}{\alpha+3\beta}}.
\end{align}
At the settings of $\alpha=5/2,$ $\beta=1/2$, we obtain
$\sigma(\hat{A})T\ge2.0\pi$, which is about five times smaller
than that obtained through Eq. \eq{variance_upper_bound}.

While maximum likelihood is a reasonable approach for
standard phase estimation, once additive errors are included,
we no longer have an unbiased estimator, so in this setting is unfair to compare
our bound to that of the Cramer-Rao bound.
Once additive errors are included, we do not have an appropriate
lower bound on the scaling.

\section{Initial Bounding Techniques}\label{app:bounding}
Our single-qubit calibration procedure works only when the errors are below a certain
initial size. Here we show how the initial size of these errors can be bounded by
conducting the appropriate experiments.

\subsection{Bounding $\epsilon$ and $\theta$}\label{sec:bound}
In Section \ref{sub:epsilon}, we showed that we can estimate
$\epsilon$ and $\theta$ at the Heisenberg limit as long as
$\epsilon^2$ and $\theta^2$ are not too large. Here, we give a
procedure to bound the initial size of $\epsilon$ and $\theta.$

Let
\begin{align}
q_0=\left|\bra{0}X_{\pi/4}(\epsilon,\theta)^4\ket{0}\right|^2.
\end{align}
By direct calculation, we have
\begin{align}\label{eq:theta_vs_epsilon}
q_0=\sin^2(\theta)+\cos^2(\theta)\sin^2\left(\frac{t\pi\epsilon}{2}\right).
\end{align}

The maximum value $\theta$ can attain is found by setting $\epsilon=0$.
This gives us
\begin{align}
|\theta|\leq \arcsin\sqrt{q_0}.
\end{align}
Likewise, the maximum value $\epsilon$ can attain is found by setting $\theta=0$.
This gives us
\begin{align}
|\epsilon|\leq \frac{2\arcsin\sqrt{q_0}}{t\pi}.
\end{align}

Now we just need to bound $q_0$. Using Hoeffding's bound,
 if we make $V$ observations
of $q_0$, we can obtain an estimate
$\hat q_0$ for $q_0$ such that
\begin{align}
 P(q_0<\hat q_0+\mu)>&1-\exp[-2V\mu^2].
\end{align}
Thus we have
\begin{align}
|\theta|\leq& \arcsin\sqrt{\hat q_0+\mu},\no
|\epsilon|\leq& \frac{2\arcsin\sqrt{\hat q_0+\mu}}{t\pi}
\end{align}
with probability $1-\exp[-2V\mu^2].$

\subsection{Bounding Measurement Error}\label{sec:meas_error}
In this section, we show how to bound $\delta_{\ketbra{0}{0},W}$ of Eq. \eq{def_delta_M},
given access to $W,$ the faulty measurement operator, and the ability
to prepare the states $\ketbra{0}{0}$ and $\varrho$ where $\varrho$
is ideally close to $\ketbra{1}{1}.$

Consider the following measurements:
\begin{align}
G_0=&\tr(W\ketbra{0}{0})\no
G_1=&\tr(W\varrho).
\end{align}
 Suppose $V$ observations are made of each variable $G_0$ and
 $G_1$, producing estimates $\widehat{G}_0$ and $\widehat{G}_1$
 of the respective variables. Then using Hoeffding's Bound,
 we have that
 \begin{align}\label{eq:Fbounds}
 P(G_0>\widehat{G}_0-\mu)>&1-\exp[-2V\mu^2]\no
 P(G_1<\widehat{G}_1+\mu)>&1-\exp[-2V\mu^2].
 \end{align}
 We will show that if
 \begin{align}\label{eq:assumptions}
 G_0>&\widehat{G}_0-\mu\equiv\widehat{G}_0^- \text{ and},\no
 G_1<&\widehat{G}_1+\mu\equiv\widehat{G}_1^+,
 \end{align}
 then
 \begin{align}
 \delta_W
 \leq& \Delta_1+\sqrt{\Delta_1^2+\Delta_2^2/2},
 \end{align}
where
  \begin{align}
\label{eq:Delta1}  \Delta_1=&\frac{(\widehat{G}_0^-)^2
-(\widehat{G}_1^+)^2-3\widehat{G}_0^--2\widehat{G}_0^-\widehat{G}_1^+-\widehat{G}_1^++2}
{2(\widehat{G}_0^--\widehat{G}_1^+)},\nonumber\\
  \Delta_2=&2(1-\widehat{G}_0^-).
  \end{align}
  By the union bound,
we have
\begin{align}
P\left(\delta_W\leq  \Delta_1+\sqrt{\Delta_1^2+\Delta_2^2/2}\right)
\geq   1-2\exp[-2V\mu^2].
\end{align}
  One can verify that if $\widehat{G}_0\approx 1$ and
  $\widehat{G}_1\approx 0$, and $\mu\ll 1$, then $\Delta_1$ and
  $\Delta_2$ are small and hence $\delta_W$ is small.

Since Pauli operators are an orthonormal basis for Hermitian
operators, we can write
\begin{align}\label{eq:pauli_basis}
W&=\sum_{i=0}^3 m_i\sop P_i\nonumber\\
\varrho&=\frac{1}{2}\left(\sop P_0+\sum_{i=1}^3r_i\sop P_i\right),
\end{align}
where $\sop P_0=\I,$ $\sop P_1=\sop P_x$, $\sop P_2=\sop P_y$,
and $\sop P_3=\sop P_z$. Additionally $W$ and $\varrho$ must be
positive semidefinite, and $0\leq\tr(W\rho)\leq 1$ for all $\rho.$

Using Eq. \eq{assumptions} and Eq. \eq{pauli_basis}, we have
\begin{align}
1\geq m_0+ m_3&>\widehat{G}_0^-,\label{eq:bound0}\\
0\leq\sum_{i=0}^3 m_ir_i&<\widehat{G}_1^+.\label{eq:bound1}
\end{align}

We will use Eq. \eq{bound0} to upper bound the size of
$ m_1$ and $ m_2.$ The eigenvalues of $W$ must lie in
the range $[0,1].$ Explicitly evaluating the eigenvalues
of $W$, and requiring that they are in this range gives
\begin{align}\label{eq:eig_bound}
0\leq m_1^2+m_2^2\leq (1-m_0)^2-m_3^2.
\end{align}
Using Eq. \eq{bound0}, we have
\begin{align}
1-m_0\geq m_3>\widehat{G}_0^--m_0.
\end{align}
Thus we can write
\begin{align}\label{eq:m_3}
m_3=f-m_0
\end{align}
for some $\widehat{G}_0^-<f\leq 1$.
Plugging Eq. \eq{m_3} into Eq. \eq{eig_bound}
and taking the derivative with respect to $m_0$,
we find
\begin{align}
0\leq m_1^2+m_2^2\leq (1-f)^2.
\end{align}
Since $\widehat{G}_0^-<f\leq 1$, we finally have
\begin{align}
0\leq m_1^2+m_2^2\leq (1-\widehat{G}_0^-)^2,
\end{align}
so
\begin{align}\label{eq:m_1bound}
|m_1|,|m_2|\leq  1-\widehat{G}_0^-.
\end{align}

Using Eq. \eq{bound0}, and that $1-r_3>0$, we have
\begin{align}
 m_3&>\frac{\widehat{G}_0^-- m_0- m_3r_3}{1-r_3}\nonumber\\
&> \frac{\widehat{G}_0^--\widehat{G}_1^+-(1-\widehat{G}_0^-)
(|r_1|+|r_2|)}{1-r_3},
\label{eq:dboundu}
\end{align}
where in the second line, we have used Eq. \eq{m_1bound}

Assuming $\widehat{G}_0^-\approx 1$ and $\widehat{G}_1^+\approx 0$, the
numerator of Eq. \eq{dboundu} will be positive. Using the positive semidefinte
constraint on $\varrho$, we have $r_3>-\sqrt{1-r_1^2-r_2^2}$,
so
\begin{align}
 m_3> \frac{\widehat{G}_0^--\widehat{G}_1^+-(1-\widehat{G}_0^-)
(|r_1|+|r_2|)}{1+\sqrt{1-r_1^2-r_2^2}}.\nonumber
\end{align}
We always want to choose $r_1=r_2.$
If $r_1\neq r_2$, we can replace $r_1$ and
$r_2$ by their average, thereby preserving the numerator
while increasing the denominator. Thus
\begin{align}\label{eq:dboundc}
 m_3> \frac{\widehat{G}_0^--\widehat{G}_1^+
-2(1-\widehat{G}_0^-)|r_1|}
{1+\sqrt{1-2r_1^2}}.
\end{align}

We now minimize the right hand side of Eq. \eq{dboundc}
with respect to $r_1$, (assuming we are in a regime where
$\widehat{G}_0^-\approx 1$, and $\widehat{G}_1^+\approx 0$,)
giving
\begin{align}
m_3>\frac{2-(2-\widehat{G}_0^-)^2-\widehat{G}_1^+(2\widehat{G}_0^--\widehat{G}_1^+)}
{2(\widehat{G}_0^--\widehat{G}_1^+)}.
\end{align}
At this point, we can bound the error that results from
using $W$ instead of the ideal $\ketbra{0}{0}.$
For an arbitrary state $\omega$ such that
\begin{align}
\omega&=\frac{1}{2}\left(\sop P_0+\sum_{i=1}^3w_i\sop P_i\right),
\end{align}
we have
\begin{align}\label{eq:Mrhodiff}
|\tr(W\omega)-\tr(\ketbra{0}{0}\omega)|
 \leq\Delta_1(1+w_3)
  + \Delta_2\sqrt{\frac{1-w_3^2}{2}}
  \end{align}
with $\Delta_1$ and $\Delta_2$ given by Eq. \eq{Delta1},
and we have used the trick of replacing $w_1$ and $w_2$ by their average.
Maximizing \eq{Mrhodiff}
   with respect to $w_3$
  we have
  \begin{align}
\left|\tr(W\omega)-\tr(\ketbra{0}{0}\omega)\right|\leq
 \Delta_1+\sqrt{\Delta_1^2+\Delta_2^2/2}.
    \end{align}
    as claimed.

\end{document}